\documentclass[a4paper,12pt]{article}
\usepackage{latexsym,amssymb,amsfonts,amsmath,amsthm}
\usepackage[dvips]{graphicx}
\usepackage{color}

\newtheorem{theorem}{Theorem}
\newtheorem{lemma}{Lemma}
\newtheorem{corollary}{Corollary}
\newtheorem{proposition}{Proposition}
\newtheorem{remark}{Remark}
\newtheorem{definition}{Definition}

\newcommand{\bs}[1]{\boldsymbol{#1}}
\newcommand{\im}{\bs{\rm i}}
\newcommand{\spann}{{\rm span}}

\newcommand{\supp}{{\rm supp}}
\newcommand{\Deg}{{\rm deg}}
\newcommand{\dist}{{\rm dist}}

\newcommand{\V}{\mathcal{V}}
\newcommand{\A}{\mathcal{A}}
\newcommand{\T}{\mathcal{T}}

\newcommand{\F}{\mathcal{F}}
\newcommand{\W}{\mathcal{W}}

\newcommand{\seg}{\textcolor{black}}
\newcommand{\hig}{\textcolor{black}}

\newcommand{\segB}{\textcolor{black}}

\title{{\Large {\bf Quantum walks induced by Dirichlet random walks on infinite trees  
}
}}
\author{ 
{\small 
Yusuke Higuchi,$^{1}$ 
\footnote{email address 
higuchi@cas.showa-u.ac.jp 
}
\quad 
Etsuo Segawa,$^{2}$ 
\footnote{e-segawa@m.tohoku@ac.jp 
}\quad
}\\ 
{\scriptsize $^{1}$ 
Mathematics Laboratories, College of Arts and Sciences, Showa University
}\\
{\scriptsize 
Fuji-Yoshida, Yamanashi 403-005, Japan
} \\
{\scriptsize $^2$ 
Graduate School of Information Sciences, Tohoku University, 
}\\
{\scriptsize 
Aoba, Sendai 980-8579, Japan
} \\
} 
\date{\empty }
\begin{document}
\maketitle

\par\noindent
\begin{small}
\par\noindent
{\bf Abstract}. 
We consider the Grover walk on infinite trees from the view point of spectral analysis. 
From the previous works, 
infinite regular trees provide localization. 
In this paper, we give the complete characterization of the eigenspace of this Grover walk, which involves localization of its behavior and recovers the previous works. 
Our result suggests that the Grover walk on infinite trees may be regarded as a limit of the quantum walk induced by the isotropic random walk with the Dirichlet boundary condition at the $n$-th depth
rather than one with the Neumann boundary condition. 


\end{small}

\setcounter{equation}{0}

\section{Introduction}
A discrete-time quantum walk on a graph $G=(V,E)$ has been proposed by Gudder (1988) as ``quantum graphic dynamics" in his book~\cite{Gud}; 
a walker jumps to neighbors with a matrix valued weight at each time step so that the time evolution of the whole state is unitary. 
This walk can be interpreted as a dynamics on the arcs of $G$~\cite{Sev}. 
From this observation, 
it is possible to naturally connect quantum walks~\cite{HKSS:YMJ,Tan} on graphs 
and the scatterings of one-dimensional plane wave on the wire~\cite{ExnSeb}. 
In this paper, we identify the discrete-time quantum walk with a pair of $U$ and $\mu$: $U$ is a unitary operator 
on $\ell^2(A):=\ell^2(A,\mathbb{C})$  with the standard inner product so that 
	\[ \langle \delta_f,  U\delta_e\rangle \neq 0 \Leftrightarrow t(e)=o(f) \]
and $\mu$ is the measurement $\mu: \ell^2(A,\mathbb{C})\to \ell^1(V,\mathbb{R}_{\geq 0})$ such that 
	\[ (\mu(\psi))(u)=\sum_{e:t(e)=u} |\psi(e)|^2, \] 
where $(\mu(\psi))(u)$ is interpreted as the findings probability at vertex $u$ if $||\psi||=1$. 
Here $A$ is the set of symmetric arcs induced by $E$, and $o(e),t(e)\in V$ are the origin and terminus of arc $e\in A$, respectively. 
We often call the unitary operator $U$ quantum walk simply.
Due to the unitarity of the time evolution $U$, we can define the distribution $\nu_n^{(\psi_0)}: V\to [0,1]$ at each time $n$ with the initial state 
$\psi_0\in \ell^2(A)$ with $||\psi_0||=1$ such that
	\[ \nu_n^{(\psi_0)} \equiv \mu(U^n\psi_0). \]
One of the main topics of the study of quantum walks is its asymptotics of $\nu_n$ for large $n$ e.g., localization, linear spreading and recurrent property and so on 
(see~\cite{Kon,CGMV} and its references). 

The Grover walk is regarded as the induced quantum walk by the underlying isotropic random walk $T$~\cite{EHSW,Sze}. 
Here $T$ is a self-adjoint operator on $\ell^2(V)$ with the standard inner product such that 
	\[\langle \delta_v,  T\delta_u\rangle=\frac{\bs{1}_{\{A\}}(u,v)}{\sqrt{\Deg (u) \Deg (v)}}. \]
Remark that $T$ is unitary equivalent to the transition operator of the isotropic random walk. 
For a {\it finite} graph $G$, the spectrum of the Grover walk $U$ is simply decomposed into 
	\[ \sigma(U)=\{e^{\pm \im \arccos \sigma(T)}\} \cup \{1\}^{b_1} \cup \{-1\}^{b_1-1+\bs{1}_{B}}, \]
where $b_1$ and $b_1-1+\bs{1}_B$ are the multiplicities of the eigenvalues $1$ and $-1$, respectively.
Here $b_1$ coincides with the first Betti number of $G$, and $\bs{1}_B = 1$ if G is bipartite, $\bs{1}_B = 0$ otherwise. 
We call the first term {\it inherited part} from the underlying RW, and the last two terms {\it birth part}. 
As is seen the above multiplicities in the birth part reflect a homological structure of the graph. 
Every eigenfunction of birth part has a finite support corresponding to fundamental cycles~\cite{HKSS:JFA}. 
On the other hand, the birth part never appears for a {\it finite} tree, because there exist no cycles.
Let us consider whether this kind of statement still holds for an {\it infinite} graphs. 
Naturally, if there exists a cycle, then we can construct the eigenfunction with a finite support along it in the same fashion as in the finite case.
This implies  localization happens if the initial state has an overlap between the birth eigenspace.
Even if there are no cycles in a graph, the localization of the Grover walk may occurs. 
In contrast, we know localization occurs on some infinite tree with some appropriate initial state~\cite{CHKS}.

We first show in this paper that for infinite tree whose minimal degree is at least $3$, 
actually there are infinite number of birth eigenfunctions (see Theorem~1). 
The birth eigenfunction is generated by a finite energy {\it flow}~\cite{Bol} starting from a vertex to the infinite down stream. 
One-dimensional lattice is also a tree with degree $\kappa=2$ and we can make an infinite flow. 
However the ``function" generated by such a flow is $\ell^\infty(A)$ but no longer $\ell^2(A)$. 
From the above reason, the Grover walk on the one-dimensional lattice never exhibit localization as is known that 
this walk is trivially one way going without any interactions. 
From the above reason, we assume that the minimal degree is at least $3$ throughout this paper. 
Once the minimal degree $\kappa_0+1\geq 3$ is assumed, then the square summability of the function generated by a flow is ensured (see Lemma~1). 
Studies of localization makes attention by many researchers. 
In the previous studies, it has been known that the derivation of localization of Grover walk 
are cycle structure~\cite{HKSS:JFA,MS} of the graph and 
inherited eigenspace from underlying isotropic random walk~\cite{HS,KOS}. 
Thus the positivity of such a geometric constant in Corollary~3 can be said to be a new criterion for the localization. 

To understand such a spectral structure of the Grover walk on infinite trees, 
we introduce a quantum walk induced by the random walk under the Dirichlet boundary condition at the $n$-depth; 
this type of random walk may be called Dirichlet random walk in this paper. 
We show that if we regard the Grover walk on the infinite tree as a limit of this induced QW for $n\to \infty$, then the birth part naturally appears. 
More precisely, the eigenfunctions of birth part of the approximate QW coincide with 
the eigenfunctions of the original Grover walk within the $n$-th depth (see Theorem~4). 
The positivity of a kind of the Cheeger constant ensures the square summability (see Corollary~3).
We note that the quantum walk induced by the Dirichlet random walk has essentially appeared 
as a tool of several quantum spatial search algorithms on graphs; 
the target vertices in these algorithms correspond to the Dirichlet cut-off ones. 
The previous works~\cite{AKR,Sze,SKW} with the Dirichlet cut-off underlying random walk for a quantum search problem focused on the inherited part 
since the ``uniform" state, which is the initial state of these algorithms, has \segB{small} overlap to the birth part \segB{for large system size}. 
On the other hand, we propose an advantage of focusing on the birth part with the Dirichlet cut-off underlying random walk 
to extract a geometric structure of infinite graphs. 

Connecting the Grover walk on the infinite tree to the quantum walks with the other construction~\cite{JoyMar} 
is one of the interesting future's problem. 

This paper is organized as follows. In Section~2, we provide the setting of the Grover walk on the infinite tree 
and introduce inherited and birth parts of the Grover walk. 
Section~3 devotes to obtain an expression of the birth eigenspace by finite energy flows on the infinite tree (Theorem~1). 
In the previous work~\cite{CHKS} (2007), it is shown that one initial state provides localization of the Grover walk on infinite regular tree, 
while the other one does not. 
\segB{Here a graph is called regular, if the degree is constant. }  
We can check that the essential difference with respect to the localization between them is the overlap to such birth eigenfunctions. 
Finally in Section~4, we introduce a quantum walk induced by the Dirichlet random walk as an approximation of the Grover walk on the infinite tree. 
We give the spectral mapping theorem from the Dirichlet random walk (Theorem~4). 
As a result, the birth part of the approximate walk can be described by that of the original Grover walk. 
\section{Grover walk on infinite tree}

\seg{To define the Grover walk on infinite tree $\mathbb{T}=(V,A)$, we prepare two boundary operators
$d_\T,d_O: \ell^2(A)\to \ell^2(V)$,  
	\begin{align} 
        (d_\T\psi)(u) &= \sum_{e: t(e)=u}1/\sqrt{\deg(t(e))}\;\psi(e), \\
        (d_O\psi)(u) &= \sum_{e: o(e)=u}1/\sqrt{\deg(o(e))}\;\psi(e),\; (u\in V) 
        \end{align} 
respectively. }
\seg{\hig{Here} $\deg(u)$ is the degree of $u\in V$. }
Putting $S: \ell^2(A)\to \ell^2(A)$ as a permutation operator denoted by 
	\[ (S\psi)(e)=\psi(\bar{e}),\] 
we have $d_O=d_\T S$. 
Here $\bar{e}$ is the inverse arc of $e$. 
\seg{The adjoint operators are 
	\begin{align}
        (d_\T^*f)(e) &=  1/\sqrt{\deg(t(e))}\; f(t(e)), \\
        (d_O^*f)(e) &=  1/\sqrt{\deg(o(e))}\; f(o(e)), \;(e\in A)
        \end{align}
respectively. }
Remark that 
	\begin{equation}\label{identity}
        d_\T d_\T^*=d_Od_O^*=\bs{1}_\V, 
        \end{equation}
where $\bs{1}_\V$ is the identity operator on $\V\equiv \ell^2(V)$. 
We define the self-adjoint operator $T$ on $\V$ by $T=d_\T d_O^{*}=d_Od_\T^{*}$; for $f\in \V$, 

This is the transition probability operator of the isotropic random walk on $\mathbb{T}$. 
\begin{definition}Grover walk on a connected tree is defined as follows: 
\begin{enumerate}
\item total space: $\A\equiv \ell^2(A)$;
\item time evolution: $U: \ell^2(A)\to \ell^2(A)$ such that 
	\[ U=S(2d_\T^{*}d_\T-\bs{1}_\A). \]
Here $\bs{1}_{\A}$ is the identity operator on $\A$. 
\end{enumerate}
\end{definition}
We read 
	\[ \langle \delta_f,U\delta_e \rangle=\bs{1}_{\{o(f)=t(e)\}}(e,f)\left( \frac{2}{\deg(t(e))}-\delta_{e,\bar{f}} \right) \]
as the amplitude associated with the one-step moving from $e$ to $f$.
We also define an important invariant subspaces $\mathcal{L}\subset \A$ and its orthogonal complement $\mathcal{L}^\bot$; 
we call them the inherited and birth subspaces, respectively. Here
	\[ \mathcal{L}=d_{\T}^*(\V)+d_O^*(\V). \]
\begin{proposition}
It holds that 
        \[ U=U|_\mathcal{L}\oplus U|_{\mathcal{L}^\bot}. \]
\end{proposition}
\begin{proof}
It is sufficient to show $U(\mathcal{L})=\mathcal{L}$. 
First we show $U(\mathcal{L})\subset \mathcal{L}$. For any $\psi\in \mathcal{L}$, there exist $f,g\in \V$ such that 
$\psi=d_{\T}^*f+d^*_Og$. By Eq.~(\ref{identity}), we have 
	\[ U\psi=-d_{\T}^* g+d_O^*(f-2Tg), \]
which implies $U\psi \in \mathcal{L}$. Conversely, it is easily checked that for any $f,g\in \V$, 
	\[ d_{\T}^*f+d_O^*g=U(d_{\T}^*(g-2Tf)-d_O^*f)\in U(\mathcal{L}). \]
\end{proof}
In our previous works~\cite{SS,HSS}, we have the general result on the spectrum of a generalized quantum walk on infinite graphs 
including the Grover walk as follows:
\begin{enumerate}
\item the spectrum $\sigma(U_{\mathcal{L}})|$ of $U$ which restricted to a subspace $\mathcal{L}$
coincides with the image of the inverse Joukowski transform of the spectrum $\sigma (T)$ of an underlying random walk:
	\[ \sigma(U|_{\mathcal{L}}) = J^{-1} (\sigma (T)), \]
where $\sigma (\cdot )$ is the set of spectrum and $J(z)=(z+z^{-1})/2$; 
\item the continuous spectrum $\sigma_{c}(U)$ of the induced quantum walk 
is completely derived from the continuous one $\sigma_{c}(T)$ of the underlying random walk: $\sigma_{c}(U)=J^{-1}(\sigma_{c}(T))$. 
\end{enumerate}
Moreover it is well known that the spectrum of the isotropic random walk $T$ on the infinite $\kappa$-regular tree 
consists of the continuous spectrum: 
	\[ \sigma(T)=\sigma_{c}(T)=[ -2\sqrt{\kappa-1}/\kappa, 2\sqrt{\kappa-1}/\kappa]. \] 
Thus after all, we can conclude that the eigenspace of the Grover walk on the infinite regular tree derives from the birth part $\mathcal{L}^{\bot}$. 
In the next section, we present a construction of this eigenspace by investigating a graph structure of the infinite tree; 
in other words, we will construct a kind of combinatorial flow with finite energy on the tree. 
\section{Construction of the eigenspace of the Grover walk on $\mathbb{T}$ }
%
We take a choice of a fixed arbitrary vertex $o$ as the root of $\mathbb{T}$. 
For each $u\in V$, we take the connected infinite subtree $\mathbb{T}^{(u)}=(V^{(u)},A^{(u)})$ which is the induced subgraph of $\mathbb{T}$ by 
	\[ V^{(u)}=\{ v\in V: \mathrm{dist}(o,v)=\mathrm{dist}(o,u)+\mathrm{dist}(u,v) \}; \]
equivalently, $\mathbb{T}^{(u)}$ is the subtree of $\mathbb{T}$ consisting of all the descendants of $u$, 
which is considered as the root of $\mathbb{T}^{(u)}$. See Figure~1. 
Here $\mathrm{dist}(x,y)$ is the usual graph-distance between two vertices $x$ and $y$. 
The set of vertices  $\{ v\in V: \mathrm{dist}(o,v)=n \}$ is often called the $n$-th depth (from $o$ ) of $\mathbb{T}$.
In particular, we define $\mathbb{T}^{(o)}=\mathbb{T}$. 
Let $\{e_0^{(u)},\dots,e_{m-1}^{(u)}\}\subset A^{(u)}$ be the set of arcs $e$ such that $o(e)=u$ and $t(e)\in V^{(u)}$,
where 
	\begin{equation}\label{mu}
        m(u)= \begin{cases}
        	\deg(u) & \text{: $u=o$,} \\
                \deg(u)-1 & \text{: $u\neq o$.}
                \end{cases}
        \end{equation}
Thus $m(u)$ is number of children of $u$. 
\begin{figure}[b]
\begin{center}
  \includegraphics[width=50mm]{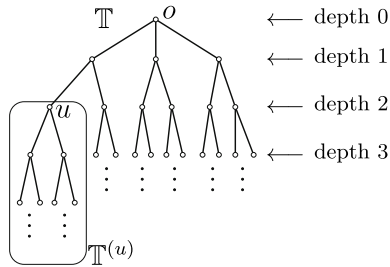}
\end{center}
\caption{$\mathbb{T}$ and $\mathbb{T}^{(u)}$. }
\label{fig:one}
\end{figure}
Now we define important functions by the recurrent way as follows, which corresponds to a kind of combinatorial flow: 
\begin{definition}\label{DefFlow}
\seg{Assume $\deg(v)\geq 2$. }
\seg{Let $\varphi^{(\pm)}: \{ (u,j); \{1,\dots, m(u)-1\} \} \to \mathbb{C}^{\A}$ be defined as follows. \\}
For $e,f\in \{g\in A^{(u)}: \mathrm{dist}(o(g),u)<\mathrm{dist}(t(g),u)\}$, $\varphi^{(\pm)}(u,j):=\varphi_{u,j}^{(\pm)}$ denotes 
	\begin{align}
        \varphi_{u,j}^{(\pm)}(e) &= 
        \begin{cases}
        \omega_j^{k}\seg{/m(o(e))} & \text{: $e=e_k^{(u)}$,\;\;$(k=0,\dots,m(u)-1)$,} \\
        \pm \varphi_{u,j}^{(\pm)}(f)/m(o(e)) & \text{: $t(f)=o(e)$, \; $o(e)\neq u$,} 
        \end{cases} \label{varphi} \\
        \varphi^{(\pm)}_{u,j}(\bar{e}) &= \mp \varphi^{(\pm)}_{u,j}(e). \label{flip} 
        \end{align} 
For $e\notin A^{(u)}$,  $\varphi^{(\pm)}_{u,j}(e)=0$. 
Here $\omega_j=e^{2\pi \im j/m(u)}$ and $\im=\sqrt{-1}$. 
\end{definition}
In the following, $\varphi^{(\pm)}_{u,j}$ denotes $\varphi^{(\pm)}(u,j)$. 
\begin{remark}
\seg{Let $\kappa_0+1$ be the minimum degree of $\mathbb{T}$. }
\seg{For $\kappa_0\geq 2$,} we have $\varphi^{(\pm)}_{u,j}\in \ell^2(A)$ as shown by the following lemma. 
\end{remark}
\begin{lemma}\label{subtreesCONS}
\seg{Assume $\kappa_0\geq 2$. 
We have 
	\[ ||\varphi^{(\epsilon)}_{u,j}||=c_{u,j}, \]
where $c_{u,j}>0$ is a uniformly bounded constant, that is, there exists $c_0<\infty$ such that $\sup_{u,j}c_{u,j}\leq c_0$. }
Moreover $\{ \varphi^{(\epsilon)}_{u,j}: u\in V, j\in{0,\dots,m(u)-1}, \epsilon\in\{\pm\} \}$ are linearly independent. 
In particular, when $\mathbb{T}$ is $\kappa$-regular tree, then they are orthogonal to each other, that is, 
for every $\epsilon,\epsilon'\in\{\pm\}$, $u,u'\in V$, $j\in\{1,\dots,m(u)-1\}$ and $j'\in\{1,\dots,m(u')-1\}$, 
	\[ \left\langle \varphi^{(\epsilon)}_{u,j}, 
        	\varphi^{(\epsilon')}_{u',j'}\right\rangle
        	=\frac{2(\kappa-1)}{m(u)(\kappa-2)} \delta_{\epsilon,\epsilon'}\delta_{u,u'}\delta_{j,j'}.  \] 
\end{lemma}
\begin{proof}
Let $\mathcal{H}^{(\pm)}$ be eigenspaces of eigenvalues $\{\pm 1\}$ of $S$, respectively, that is, 
	\begin{equation}
        \mathcal{H}^{(\pm)}=\{ \psi\in \A: \psi(\bar{e})=\pm \psi(e)\;\;(e\in A) \}. 
        \end{equation}
First we show that the linear independence. 
\segB{Assume that $\varphi^{(+)}_{u,j}$ is expressed by a linear combination of $\varphi_{u',j'}^{(\pm)}$'s with 
$u'\in V,\;j'\in\{1,\dots,m(u')-1\}$: 
	\begin{equation}\label{LinComb}
        \varphi^{(+)}_{u,j}=\sum_{(u',j')\neq (u,j)} r^{(+)}_{u',j'}\varphi^{(+)}_{u',j'} + \sum_{(u',j')} r^{(-)}_{u',j'}\varphi^{(-)}_{u',j'}.
        \end{equation}
We will show the contradiction step by step as follows: 
}
\begin{enumerate}
\item 
From the definition, we have $\varphi^{(+)}_{u,j}\in \mathcal{H}^{(-)}$ and $\varphi^{(-)}_{u',j'}\in \mathcal{H}^{(+)}$. 
Thus $\varphi^{(+)}_{u,j}\bot \varphi^{(-)}_{u',j'}$. 
\segB{Therefore, any $\varphi^{(-)}_{u',j'}$'s with $u'\in V,\;j'\in\{1,\dots,m(u')-1\}$ are not contained in the linear combination (\ref{LinComb}), that is, 
$r_{u',j'}^{(-)}=0$ for every $u'\in V,\;j'\in \{1,\dots,m(u')-1\}$. }
\item 
Note that 
if $V^{(u)}\cap V^{(u')}=\emptyset$, then $\varphi^{(\epsilon)}_{u,j}\bot \varphi^{(\epsilon)}_{u',j'}$ holds since their supports are disjoint, 
where the support of $\psi\in \A$ is denoted by $\mathrm{supp}(\psi)=\{a\in A : \psi(a)\neq 0\}$. 
\segB{Therefore, any $\varphi^{(+)}_{u',j'}$'s with $V^{(u)}\cap V^{(u')}=\emptyset$ are not contained in the linear combination (\ref{LinComb}), that is, 
$r_{u',j'}^{(+)}=0$ for every $u'$ with $V^{(u)}\cap V^{(u')}=\emptyset$, $j'\in\{1,\dots,m(u')-1\}$. }
\item 
Remark that $V^{(u)}\cap V^{(u')}\neq \emptyset$
if and only if $\mathbb{T}^{(u)}\supseteq \mathbb{T}^{(u')}$ or $\mathbb{T}^{(u)}\subseteq \mathbb{T}^{(u')}$. 
If $\mathbb{T}^{(u)}\supsetneq \mathbb{T}^{(u')}$, then 
\segB{
\[ \mathrm{supp}(\varphi_{u,j}) \setminus \mathrm{supp}(\varphi_{u',j'})\supset \{a\in A^{(u)} : t(a)=u,\;\mathrm{or}\;o(a)=u \}. \]
Therefore, $\varphi^{(+)}_{u,j}$ cannot be described by the linear combination only using
	\[ \{\varphi_{u',j'}^{(+)}: u'\in V^{(u)}\setminus\{u\},\;j'\in\{1,\dots,m(u')-1\}\}. \] }
\item 
\segB{The shortest path from $o$ to $u$ is denoted by $(e_1,e_2,\dots,e_n)$ with $o(e_1)=o$ and $t(e_n)=u$. 
By (1)-(3), we have 
	\[ \varphi_{u,j}^{(+)}=\sum_{(u',j')\in Q} r_{u',j'}^{(+)}\varphi_{u',j'}^{(+)} \]
Here 	
	\[ Q= \{(u',j') \;:\; u'\in \{o,t(e_1),\dots,t(e_{n-1})\} \cup V^{(u)},\;j'\in\{1,\dots,m(u')-1\}  \}. \]
Assume that $\varphi_{o,j'}^{(+)}$ is contained in the linear combination. 
However, the values of the linear combination function on all arcs $e$ with $o(e)=o$ or $t(e)=o$ come from only $\varphi_{o,j'}^{(+)}$. 
Therefore any $\varphi_{o,j'}^{(+)}$ $(j'\in\{1,\dots,m(o)-1\})$ are not contained in the linear combination, that is, $r_{o,j'}$'s are $0$. 
By taking the same argument recursively, we conclude that any 
$\varphi_{u',j'}^{(+)}$'s with $u'\in\{o,t(e_1),\dots,t(e_{n-1})\},\;j'\in \{1,\dots,m(u')-1\}$ are not contained in the linear combination (\ref{LinComb}).
Thus 
	\[ \varphi_{u,j}^{(+)}
        =\sum_{j'\neq j} r_{u,j'}^{(+)}\varphi_{u,j'}^{(+)}
        +\sum_{{\scriptsize \begin{matrix}u' \in V^{(u)}\setminus\{u\} \\ j' \in\{1,\dots,m(u')-1\}\end{matrix}}} r_{u',j'}^{(+)}\varphi_{u',j'}^{(+)}. \] 
}  
\item For $u=u'$, 
since ${}^T[1,w_j,\dots,w_j^{m(u')-1}]$ and ${}^T[1,w_{j'},\dots,w_{j'}^{m(u')-1}]$ are orthogonal to each other $(j\neq j')$ with respect to
$\mathbb{C}^{m(u')}$-inner product, then $\varphi_{u,j}^{(+)}$ and $\varphi_{u,j'}^{(+)}$ are linearly independent. 
\segB{Therefore $r_{u,j'}$'s with $j'\in\{1,\dots,m(u)-1\}\setminus \{j\}$ are $0$ which implies the contradiction to the statement of (3). }
\end{enumerate}
Then we have shown the linearly independence of $\{\varphi_{u,j}^{(\pm)}:u\in V,j\in\{1,\dots,m(u)-1\}\}$. 


Next we show if $k_0\geq 2$, then $\varphi:=\varphi_{u,j}^{(\epsilon)}\in \A$. 
Let $\partial A^{(u)}_{n}$ be the set of $n'$-th depth of arcs in $\mathbb{T}^{(u)}$, that is, 
	\[ \partial A^{(u)}_{n}=\{e\in A^{(u)}: \mathrm{dist}(u,o(e))=n, \;\mathrm{dist}(u,t(e))=n+1\}. \]
We have $|| \varphi|_{\partial A^{(u)}_0} ||^2= 2/m(u)\leq 2/k_0$, and 
	\[ ||\; \varphi|_{\partial A^{(u)}_1}\; ||^2= 2\times \sum_{k=0}^{m(u)-1}\frac{1}{m^2(u)m^2(t(e^{(u)}_k))}\leq 2/k_0^3.  \]
Here for any subset $A'\subset A$ and $\psi\in\A$, $\psi_{A'}$ denotes 
	\[ (\psi|_{A'})(e)=\begin{cases} \psi(e) & \text{: $e\in A'$} \\ 0 & \text{: $e\notin A'$.}  \end{cases} \]
We set $(u,v_{1}^{(w)},\dots,v_{n}^{(w)},w)$ as the shortest path between $u$ and $w$ for $w\in A^{(u)}$. 
Now we show that $||\; \varphi|_{\partial A^{(u)}_{s}}\;||^2\leq 2 /k_0^{s+2}$ by induction with respect to $s$: 
	\begin{align*} 
        ||\varphi|_{\partial A^{(u)}_{s+1}}||^2 &= \sum_{x\in t(\partial A^{(u)}_{s+1})}\prod_{i=1}^{s+1}\frac{1}{m^2(u)m^2(v_{i}^{(x)})} \\
        	&= \sum_{y\in t(\partial A^{(u)}_{s})}\frac{1}{m(y)}\prod_{i=1}^{s}\frac{1}{m^2(u)m^2(v_{i}^{(y)})} \\
                &\leq \frac{1}{k_0}\sum_{y\in t(\partial A^{(u)}_{s})}\prod_{i=1}^{s}\frac{1}{m^2(u)m^2(v_{i}^{(y)})}
                =\frac{1}{k_0}||\varphi|_{\partial A^{(u)}_{s}}||^2\leq \frac{1}{k_0} \times  \frac{2}{k_0^{s+2}}
        \end{align*}
Then we have 
	\begin{align}
        || \varphi ||^2\leq \sum_{s=0}^{\infty}  \frac{2}{k_0^{s+2}} = \frac{2}{k_0(k_0-1)}<\infty.
        \end{align}

Finally, under the assumption that $\mathbb{T}$ is the $\kappa$-regular tree, we prove the orthogonality. 
We put $\varphi^{(\epsilon)}_{u,j}|_{A^{(u')}}=:\phi$ and 
$\varphi^{(\epsilon)}_{u',j'}=:\varphi'$. 
Remark that 
	\[ \langle\varphi^{(\epsilon)}_{u,j}, \varphi^{(\epsilon)}_{u',j'}\rangle=\langle \phi,\varphi' \rangle. \]
Since the orthogonality of $\epsilon\neq \epsilon'$ case has been already shown, so we consider $\epsilon= \epsilon'$ case. 
We should remark that for any $e\in \partial A_{n'}^{(u')}$ there uniquely exists $k\in \{0,\dots,m(u)-1\}$ such that 
	\begin{equation}\label{01}
        \seg{ \phi(e)=\frac{\omega_j^k}{m(u)(\kappa-1)^{n+n'}}. }
        \end{equation}
\seg{By (\ref{01}), then if $u\neq u'$, } 
	\begin{equation}\label{a0}
        \sum_{e\in \partial A^{(u')}_{n'}} \varphi'(e)\overline{\phi(e)} 
        = \frac{\omega_j^{-k}}{m(u)(\kappa-1)^{n+n'}}\sum_{e\in \partial A^{(u')}_{n'}} \varphi'(e)=0
        \end{equation}
Here the fact $\sum_{k=0}^{m(v)-1} w_{j'}^{k}=0$ leads the second equality. 
If $u=u'$ case, we have 
	\begin{equation}\label{a1}
        \sum_{e\in \partial A^{(u')}_{n'}} \varphi'(e)\overline{\phi(e)}
        	= \left( \frac{1}{\kappa-1} \right)^{2n'} \sum_{k=0}^{m(u')-1} \frac{w_{j'}^{k}\bar{w}_{j}^{k}}{m^2(u')}  
                	= \frac{\delta_{j'j}}{m(u')}\left(\frac{1}{\kappa-1}\right)^{2n'}. 
        \end{equation}
Therefore (\ref{a0}) and (\ref{a1}) imply if $j \neq j'$, then 
	\begin{equation}\label{02}
        \sum_{e\in \partial A^{(u')}_{n'}}\varphi'(e)\overline{\phi(e)}=\sum_{e\in \partial A^{(u')}_{n'}}\overline{\phi(\bar{e})}\varphi'(\bar{e})=0,\;\;(n'\geq 0). 
        \end{equation}
The second equality derives from (\ref{flip}). 
Then if $\phi\neq \varphi'$, we have  
	\begin{equation}
        \langle \phi,\varphi' \rangle
        	= \sum_{n'=0}^{\infty} \sum_{e\in \partial A^{(u')}_{n'}}(\overline{\phi(e)}\varphi'(e)+\overline{\phi(\bar{e})}\varphi'(\bar{e}))=0.
        \end{equation}
On the other hand, if $\phi=\varphi'$, that is, $\varphi_{u,j}^{(\epsilon)}=\varphi_{u',j'}^{(\epsilon)}$ then from (\ref{a1}), 
	\begin{equation}
        ||\varphi'||^2=2\sum_{n'=0}^{\infty}\sum_{e\in \partial A^{(u')}_{n'}}\frac{1}{m(u')}\left(\frac{1}{\kappa-1}\right)^{2n'}=\frac{2(\kappa-1)}{m(u')(\kappa-2)}.
        \end{equation}        
This completes the proof.
\end{proof}

We set the following subspace of $\A$ spanned by ${\varphi^{(\pm)}_{u,j}}$' as 
	\begin{equation}\label{flowES} 
        \mathcal{F}^{(\pm)}=\spann\left\{ \varphi^{(\pm)}_{u,j}: u\in V, \;j\in \{1,\dots,m(u)-1\} \right\}, 
        \end{equation}
which is the finite span of $\{\varphi^{(\pm)}_{u,j}\}$'s, that is,  
	\[ \mathcal{F}^{(\pm)}=
        \left\{\sum_{u\in V'}\sum_{j=1}^{m(u)-1}c_{u,j}\varphi^{(\pm)}_{u,j}: c_{u,j}\in \mathbb{C},\;|V'|<\infty\right\}. \]
By Lemma~\ref{subtreesCONS}, 
$\left\{ \varphi^{(\pm)}_{u,j}/\seg{\sqrt{c_0}}\;:\; u\in V,\;j\in\{1,\dots, m(u)-1\} \right\}$ is a complete system of 
$\overline{\mathcal{F}^{(\pm)}}$. In particular, if $\mathbb{T}$ is $\kappa$-regular tree with $\kappa\geq 3$, then 
it becomes a complete orthogonal normalized system. 
\seg{Here for $\mathcal{K}\subset \A$, $\overline{\mathcal{K}}$ is the closed linear span of $\mathcal{K}$ such that
	\[ \{ \psi\in\A : \forall \epsilon>0,\; \exists \varphi\in \mathcal{K}\;s.t.,\; ||\psi-\varphi||<\epsilon \}, \]
where $||\cdot||$ is the $\ell^2$-norm. }
%
\begin{lemma}
	\begin{equation}\label{1005}
        \ker(d_\T)\cap \mathcal{H}^{(\pm)}=\ker(d_O)\cap \mathcal{H}^{(\pm)}
        \end{equation}
        and 
	\begin{equation}
        \mathcal{L}^\bot = \left(\ker(d_O)\cap \mathcal{H}^{(+)}\right) \oplus \left(\ker(d_O)\cap \mathcal{H}^{(-)}\right). 
        \end{equation}
	Under this decomposition we have 
        \begin{equation}\label{ibu}
        \seg{U|_{\mathcal{L}^\perp}= -1 \oplus 1. }
        \end{equation}
\end{lemma}
\begin{proof}
For any $\psi\in \mathcal{H}^{(\pm)}$, we have $d_O\psi=d_\T S\psi=\pm d_T\psi$, which implies 
$\ker(d_\T)\cap \mathcal{H}^{(\pm)}=\ker(d_O)\cap \mathcal{H}^{(\pm)}$. 
Since $\mathcal{L}^\bot=\ker(d_\T)\cap \ker(d_O)$ and $\mathcal{H}^{(+)}\oplus \mathcal{H}^{(-)}=\A$, 
	\begin{align*}
        \mathcal{L}^\bot &= \left(\ker(d_\T)\cap\ker(d_O)\cap \mathcal{H}^{(+)}\right) \oplus \left(\ker(d_O)\cap \ker(d_\T)\cap \mathcal{H}^{(-)}\right), \\
        	&=\left(\ker(d_O)\cap\mathcal{H}^{(+)}\right) \oplus \left(\ker(d_O)\cap \mathcal{H}^{(-)}\right)
        \end{align*}
which follows from (\ref{1005}). 
For any $\psi\in\ker d_{\T}$, it holds that $U\psi=SC\psi-S\psi$, Then we immediately obtain (\ref{ibu}). 
\end{proof}
It follows from Lemma~1 that $\mathcal{F}^{(+)}$ and $\mathcal{F}^{(-)}$ are orthogonal. 
\begin{theorem}\label{thm1}
Assume $\mathbb{T}$ is a tree with $\kappa_0\geq 2$. 
Then $\{\pm 1\}\subset \sigma(U|_{\mathcal{L}^\perp})$, which are eigenvalues. 
Moreover the birth eigenspace $\mathcal{L}^\perp$ is the infinite dimensional subspace of $\ell^2(A)$ which can be expressed as 
	\begin{align}
        \mathcal{L}^\bot 
        	= \overline{\mathcal{F}^{(+)}}\oplus \overline{\mathcal{F}^{(-)}}. 
	\end{align}
Here $\mathcal{F}^{(\pm)}$ is (\ref{flowES}). 
\end{theorem}
\begin{remark}
If $\mathbb{T}=\mathbb{Z}$, then $\mathcal{L}^\perp = \emptyset$. 
\end{remark}
\begin{proof}
It is sufficient to show that
$\overline{\mathcal{F}^{(+)}}=\ker(d_O)\cap \mathcal{H}^{(-)}$ and $\overline{\mathcal{F}^{(-)}}=\ker(d_O)\cap \mathcal{H}^{(+)}$. 

First we show $\overline{\mathcal{F}^{(\pm)}} \subset \ker(d_O)\cap \mathcal{H}^{(\mp)}$. 
\seg{By definition of $\mathcal{F}^{(\pm)}$, it is obvious that 
	\[ \mathcal{F}^{(\pm)} \subset \ker(d_O)\cap \mathcal{H}^{(\mp)}. \]
Since $\mathcal{H}^{(\pm)}$ and $\ker d_O$ are closed sets, we have $\overline{\mathcal{F}^{(\pm)}}\subset \ker(d_O)\cap \mathcal{H}^{(\mp)}$.}

Now we will show $\overline{\mathcal{F}^{(+)}} \supset \ker(d_O)\cap \mathcal{H}^{(-)}$. 
We define a subtree $\mathbb{T}_\psi=(V_\psi,A_\psi)$ induced by $\psi\in \ker(d_O)\cap \mathcal{H}^{(-)}$ as 
	\[ V_\psi=\{o(e)\in V: e\in \supp(\psi)\},\;A_\psi=\supp(\psi),  \] 
where $\supp(\psi)=\{e\in A: \psi(e)\neq 0\}$ for $\psi \in \A$. 
Since $\psi\in\ker(d_O)$, for any $e\in \supp(\psi)$, we have $\deg(o(e))\geq 2$, 
which means that $\mathbb{T}_\psi$ is decomposed into disjoint infinite connected subtrees which have no leaves. 
Then $\psi$ can be decomposed into $\psi=\psi_1\oplus \psi_2\oplus\cdots$. 
Here $\supp(\psi_i)\cap \supp(\psi_j)=\emptyset$ $(i\neq j)$, and 
for any $e\in\supp(\psi_j)$, there exists at least one arc $e'\in \supp(\psi_j)\setminus\{\bar{e}\}$ 
such that $t(e)=o(e')$. 

From now on, we assume $\psi\in \ker(d_O)\cap \mathcal{H}^{(-)}$ so that $\mathbb{T}_\psi$ is connected. 
We put $o_\psi\in V_\psi$ as the most closest vertex from the origin vertex $o$. 
We define $e^{(u)}_k$ $(k=0,\dots,m(u)-1)$ are defined by 
$o(e^{(u)}_k)=u$ with $\dist(o,o(e^{(u)}_k))<\dist(o,t(e^{(u)}_k))$, 
where $m(u)$ is defined by Eq.~(\ref{mu})\segB{: in other words, they are arcs from $u$ to its children.} 
For every $u\in V_\psi\setminus \{o_\psi\}$, 
\hig{
we also define $e^{(u)}_{\hig{-}}\in A_\psi$ by the arc such that 
$o(e_{\hig{-}}^{(u)})=u$ with $\dist(o,o(e_{\hig{-}}^{(u)}))>\dist(o,t(e_{\hig{-}}^{(u)}))$: in other words, it is the one from $u$ to its parent.  
}
For each $u\in V_\psi$, we set subspaces $\A_u$, $\V_u$ and $\W_u$ by 
\begin{align*}
	\A_u &= \{\phi\in \A : o(e)\neq u \Rightarrow \phi(e)=0 \}, \\
        \V_u &= \begin{cases}
        	\bs{0} & \text{: $u=o_{\hig{\psi}}$, } \\
                \\
        	\left\{ \phi \in \A_u: \phi(e_0^{(u)})=\cdots=\phi(e_{\kappa-1}^{(u)})=-\phi(\hig{e_{-}^{(u)}})/m(u) \right\} & \text{: $u\neq o_{\hig{\psi}}$.}
                \end{cases}\\
        \W_u &= \begin{cases}
        	\bigg\{ \phi\in \A_u: \sum_{k=0}^{m(u)-1}\phi(e_k^{(u)})=0 \bigg\} & \text{: $u=o_\psi$,} \\
                \\
        	\bigg\{ \phi\in \A_u: \phi(\hig{e_{-}^{(u)}})=\sum_{k=0}^{m(u)-1}\phi(e_k^{(u)})=0 \bigg\} & \text{: $u\neq o_\psi$, }
                \end{cases} 
        \end{align*}
where $\bs{0}$ is the set whose element is only $0$-constant function.
%
\begin{remark}
Put $\mathcal{X}_u=\{\phi\in \mathcal{A}: \phi(e_-^{(u)})=\phi(e_0^{(u)})=\cdots=\phi(e_{m(u)-1}^{(u)}) \}$. 
Then by definition, 
	\[ \A_u=\V_u\oplus \W_u \oplus \mathcal{X}_u.  \]
Moreover, for any $\phi\in \V_u\oplus \W_u$, we have $\sum_{e:o(e)=u}\phi(e)=0$ which means $\phi\in \ker(d_O)$ while 
for any $\phi'\in \mathcal{X}_u$, it is obvious that $\phi'\in \ker(d_O)^\bot$.  
\end{remark}
Consider $\mathbb{T}^{(o_\psi)}=(V^{(o_\psi)},A^{(o_\psi)})$, 
which is the subtree of $\mathbb{T}$ with 
the root $o_\psi$.  Then $\mathbb{T}_{\psi}$ is a subtree of it. 
Moreover we set 
$V_{i}=
\{v\in V^{(o_\psi)}; \text{dist}(o_{\psi},v)=i \text{ in }\mathbb{T}^{(o_\psi)}\}$
 and 
$ A_{i}=\{e\in A^{(o_\psi)}; o(e)\in V_{i}\text{ and }t(e)\in V_{i+1}\}$ for 
$i=0,1,2,\dots$. Naturally $V_{0}=\{o_\psi\}$ and 
\segB{$A_{0}=\{e\in A_{o_\psi}; o(e)=o_\psi, \;t(e)\in V_1\}$. }

Since $\ker(d_O)\cap \A_u=\V_u\oplus\W_u$ from the above remark, 
we have
	\begin{equation}
        \psi|_{\A_u}=f_u\oplus g_u,
        \end{equation}
where $f_u\in \V_u$ and $g_u\in \W_u$. 
\segB{Here $\psi|_{\A_u}$ is the projection of $\psi$ onto $\A_u$, that is, 
	\[ \psi|_{\A_u}(a)=\begin{cases} \psi(a) & \text{: $o(a)=u$,} \\ 0 & \text{: otherwise.}  \end{cases} \;(u\in V)\]}
\hig{
Let us construct a sequence $\{\phi_{n}\}$ for approximating 
$\psi \in \ker(d_O)\cap \A_u$. }

\hig{
For $u=o_{\psi}$, we have $\psi|_{\mathcal{A}_{o_{\psi}}}=g_{o_{\psi}}$. 
Thus there exists a unique system of complex coefficients $\{C_{o_{\psi},j}\}$ such 
that 
$$
g_{o_{\psi}} = 
\sum_{j=1}^{ m(o_{\psi})-1}
 C_{o_{\psi},j}\cdot\varphi^{(\pm)}_{o_{\psi},j}|_{\mathcal{A}_{o_{\psi}}};
$$
we set $$
\phi_{0}=\sum_{j=1}^{m(o_{\psi})-1}C_{o_{\psi},j}\cdot\varphi^{(\pm)}_{o_{\psi},j}. 
$$
}
\hig{
We should remark that $\psi(e)=\phi_{0}(e)$ for 
every $e\in A\backslash \cup_{i=1}^{\infty}(A_{i}\cup\bar{A_{i}})$. 
Here $\bar{E}=\{e\, |\,\bar{e}\in E\}$ for the set of arcs $E$. 
Moreover, if there exists an arc $e_{0}\in A_{0}$ such that $\psi(e_{0})=0$, 
then $\phi(e)=0$ for every arc $e$ in $\mathbb{T}^{t(e_{0})}$, since 
$\mathbb{T}_{\psi}$ is assumed to be connected. On the other hand, 
we can easily check $\phi_{0}(e)=0 $ for every arc $e$ in $\mathbb{T}^{t(e_{0})}$ 
by our construction in Definition~2. 
We have $\psi = \phi_{0}$ in $\mathbb{T}^{t(e_{0})}$.  
}

\hig{
Next we focus on $V_{1}$ and $A_{1}$. From our observation stated above, 
we only have to treat $V_{1}\cap \{t(e)\, ;\, \psi(e)\not=0
\text{ for } e\in A_{0}\}$. 
It follows from Remark~2 that, for $u=t(e)\in V_{1}$ such that $e\in A_{0}$,
$$
\psi|_{\mathcal{A}_{u}}=f_{u}\oplus g_{u};
$$
we can set $f_{u}=\phi_{0}|_{\mathcal{A}_{u}}$, since $f_u$ is the unique function in $\V_u\oplus \W_u$ such that $f_u(\bar{e})\neq 0$. 
Obviously it holds that 
$\sum_{e\in A_{u}}\psi(e)=\sum_{e\in A_{u}}\phi_{0}(e)$. 
Remark again we assume $\psi(e)\not=0$ for $e\in A_{0}$ and $t(e)=u$ here. 
Then there exists a unique system of complex coefficients $\{C_{u,j}\}$ such 
that 
$$
g_{u} = 
\sum_{j=1}^{ m(u)-1}
 C_{u,j}\cdot\varphi^{(\pm)}_{u,j}|_{\mathcal{A}_{u}};
$$
we set 
$$
\phi_{1}=\phi_{0} + \sum_{u\in V_{1}\cap \{t(e)\, ;\, \psi(e)\not=0
\text{ for } e\in A_{0}\}}
\sum_{j=1}^{m(u)-1}C_{u,j}\cdot\varphi^{(\pm)}_{u,j}. 
$$
}

\hig{Recursively we can set, for any $n=1,2,\dots ,$ 
$$
\phi_{n}=\phi_{n-1} + \sum_{u\in V_{n}\cap \{t(e)\, ;\, \psi(e)\not=0
\text{ for } e\in A_{n-1}\}}
\sum_{j=1}^{m(u)-1}C_{u,j}\cdot\varphi^{(\pm)}_{u,j}
$$
as in the same fashion stated above. 
It is easy to see that $\psi(e) = \phi_{n}(e)$, 
$e\in A\backslash \cup_{i=n+1}^{\infty}(A_{i}\cup\bar{A_{i}})$ and that 
$\phi_{n}(e)\in \mathcal{F}^{(+)}$ for every $n=1,2,\dots$ . 
}

\hig{
Now let us show $\psi\in\bar{\mathcal{F}^{(+)}}$. 
Let $\psi_{>m}=\psi|_{\cup_{i=m+1}^{\infty}(A_{i}\cup\bar{A_{i}})}$ and 
 $\phi_{>m}=\phi_{m}|_{\cup_{i=m+1}^{\infty}(A_{i}\cup\bar{A_{i}})}$. 
Since $\psi\in \A=\ell^{2}(A)$ by the assumption, for any $\epsilon>0$, 
there exists $k$, such that 
	\[ || \psi_{>k} ||^2 <\epsilon. \]
On the other hand, it holds that	
	\begin{align}
        \| \psi-\phi_k \|^{2} &= \|\psi_{>k} -\phi_{>k} \|^{2}
                \leq 2(\|\psi_{>k}\|^{2} +\|\phi_{>k} \|^{2})\\
                &\leq 2\big(\epsilon + 
                \seg{\sum_{e\in A_{k}} \frac{|\psi(e)|^{2}}{m^2(t(e))} || \varphi^{(+)}_{t(e),j} ||^2} \big)\\
                &\leq 2\epsilon\left(1+\frac{2c_0}{\kappa_0^2}\right).
      	\end{align}
}
Then we have $\psi \in \overline{\F^{(+)}}$. 

When $\mathbb{T}_\psi$ is a union of disjoint infinite subtrees, we take a linear combination of the above $\psi$'s 
and take a similar estimation.
Then we have $\psi\in \overline{\F^{(+)}}$. 
Now we arrive at $\ker(d_O)\cap \mathcal{H}^{(-)} \subset \overline{\F^{(+)}}$. 
In a similar way, we obtain $\ker(d_O)\cap \mathcal{H}^{(+)} \subset \overline{\F^{(-)}}$. 
\end{proof}

We close this section to illustrate examples.  
\seg{The first one recovers some previous results in \cite{CHKS}.}
We define $\bar{\nu}^{(J)}: V(\mathbb{T}_\infty) \to [0,1]$ $(J\in\{A,B\})$ by 
	\[ \bar{\nu}^{(J)}(u):=\lim_{T\to\infty}\frac{1}{T}\sum_{n=0}^{T-1}\left\{\sum_{e:t(e)=u} |(U^n\psi_0^{(J)})(e)|^2\right\}.  \]
\seg{\begin{corollary}
If $\sigma(T)=\sigma_c(T)$, then 
	\begin{align}
	 \bar{\nu}^{(J)}(u)=\sum_{e:t(e)=u}\left|\left(\Pi_{\mathcal{F}_+\oplus \mathcal{F}_-}\psi_0^{(J)}\right)(e)\right|^2.
	\end{align}
\end{corollary}
In particular, if $\mathbb{T}$ is $\kappa$-regular, for the following initial states introduced by \cite{CHKS} 
	\begin{align}
        (\psi_0^{(A)})(e) &= \begin{cases} 1 & \text{: $e=e_k$, ($k=0,\dots,\kappa-1$),}\\ 0 & \text{: otherwise}. \end{cases} \\
        (\psi_0^{(B)})(e) &= \begin{cases} e^{2\pi \im k/\kappa} & \text{: $e=e_k$,  ($k=0,\dots,\kappa-1$),}\\ 0 & \text{: otherwise}, \end{cases} 
        \end{align}
we have 
	\[ \bar{\nu}^{(J)}(u)
        	=\begin{cases} 
                0 & \text{: $J=A$, }\\ 
        	\frac{(\kappa-2)^2}{2}\left(\frac{1}{\kappa-1}\right)^{3\mathrm{dist(o,u)}+1-\delta_o(u)} & \text{: $J=B$.} 
         	\end{cases}
        \]
which agrees with the previous result on \cite{CHKS}. }
The essential difference between initial states $\psi_0^{(A)}$ and $\psi_0^{(B)}$ is that $\psi_0^{(A)}\in {\mathcal{F}_\pm^{(\infty)}}^{\bot}$ 
while $\psi_0^{(B)}\notin {\mathcal{F}_\pm^{(\infty)}}^{\bot}$. 

\hig{
Let us give another type of example, in which we show the spectrum 
$\sigma (U)$ of $U$ consists of only eigenvalues for a special class of trees.
In other words, only localization of Grover walk 
occurs on such trees. 
As is seen in Sect.~2, if $\mathbb{T}$ is the $\kappa$-regular tree, 
$\sigma(U)$ can be expressed as 
\[
\sigma(U) = J^{-1}(\sigma(T))\cup\{\pm 1\},
\]
where $\sigma(T) =\sigma_{c}(T)=[-2\sqrt{\kappa -1}/\kappa,2\sqrt{\kappa -1}/\kappa]$, $J(z) = (z+z^{-1})/2$ and each of eigenvalues $\pm 1$ has 
infinite multiplicity. It is obvious to see 
$\sigma_{c}(U)=J^{-1}(\sigma(T))\not=\emptyset$. 
}

\hig{
Now we put, 
for the root $o$ of a tree $\mathbb{T}$, 
\[
K(r) = \inf \{\deg (x)\, ;\, x\in V(\mathbb{T}),\mathrm{dist}(o,x)\geq r\}.
\]
Obviously $\lim_{r\to\infty}K(r)=\kappa<\infty$ 
for the $\kappa$-regular tree. 
On the other hand, 
if $\mathbb{T}$ is a tree with 
$\lim_{r\to\infty}K(r) = \infty$, then it is called {\it rapidly branching}. 
For the discrete Laplacian, equivalently, the transition operator $T$ for 
isotropic random walks, it is shown in \cite{Fuj} 
that any rapidly branching tree has no continuous spectrum:
\begin{theorem}[\cite{Fuj}]
Let $\mathbb{T}$ be a rapidly branching tree. 
Then the continuous spectrum $\sigma_{c}(T)=\emptyset$, 
the essential one of $\sigma(T)$ coincides $\{0\}$ and 
\segB{every $\lambda\in \sigma(T)\setminus\{0\}$} is an eigenvalue with finite multiplicity. 
Moreover, if $0$ is an eigenvalue, then it has infinite multiplicity.
\end{theorem}
For a kind of family of trees, it can be decided whether $0$ is an eigenvalue 
or not. For details, refer to \cite{Fuj}. 
Combining the theorem above, the spectral mapping theorem with $J(z)$ and 
Theorem~\ref{thm1}, we can easily obtain the following:
\begin{corollary}
Let $\mathbb{T}$ be a rapidly branching tree. 
Then $\sigma(U)$ consists of only eigenvalues and their accumulating points. 
In addition, if $\lambda\in\sigma(U)\backslash\{\pm 1,\pm \im\}$, then 
it is an eigenvalue of finite multiplicity; 
if $\lambda=\pm 1$, then 
it is an eigenvalue of infinite multiplicity; 
if $\lambda=\pm\im$, then it is the accumulating point of eigenvalues or 
an eigenvalues of infinite multiplicity. 
\end{corollary}
}
\section{An approximation of the Grover walk on $\mathbb{T}$}
In this section, we consider a quantum walk induced by a random walk with the Dirichlet boundary condition. 
Almost all topics in this paper are discussed on trees, but the concept of this quantum walk is defined not only for a tree 
but for a general graph. Thus we firstly in Sect.~4.1 construct such a quantum walk on a general graph, 
which applied to a quantum search algorithm with some marked elements on it \cite{AKR,Sze,SKW}. 
After that, in Sect.~4.2, we return to the case where a graph is a tree and a quantum walk is the Grover walk on $\mathbb{T}$. 
This observation gives some information on the structure of the birth eigenspace discussed in Theorem~1. 
\subsection{QW induced by Dirichlet random walk on graphs}
Let $G'=(V',A')$ be a connected graph, and $M\subset V'$ be the set of marked elements. 
The total Hilbert space of this quantum walk is $\mathcal{A}'=\ell^2(A')$ and  
we set $\mathcal{V}'$ by $\ell^2(V')$. 
We assign a weight to each arc $\alpha: A'\to \mathbb{C}$ so that $\sum_{a:o(a)=v}|\alpha(a)|^2=1$ for all $v\in V'$. 
Define $d_{\T,M},d_{O,M}: \A'\to \V'$ by 
	\begin{align*}
         (d_{\T,M}\psi)(v) &= 
        	\begin{cases}
                \sum_{e:t(e)=v}\alpha(\bar{e}) \psi(e) & \text{: $v\in M^c$,}\\
                0 & \text{: $v\in M$. }
                \end{cases} \\
        (d_{O,M}\psi)(v) &= 
        	\begin{cases}
                \sum_{e:o(e)=v}\alpha(e) \psi(e) & \text{: $v\in M^c$,}\\
                0 & \text{: $v\in M$. }
                \end{cases}
        \end{align*} 
We define $A_M^{(+)}\subset A'$ by $\{e\in A': t(e)\in M^c\}$ and $A_M^{(-)}\subset A'$ by $\{e\in A': o(e)\in M^c\}$. 
The adjoint operators are
	\begin{align*}
         (d_{\T,M}^*f)(e) &= 
        	\begin{cases}
                \overline{\alpha(\bar{e})} f(t(e)) & \text{: $e\in A_M^{(+)}$,}\\
                0 & \text{: $e\notin A_M^{(+)}$. }
                \end{cases} \\
        (d_{O,M}^*f)(e) &= 
        	\begin{cases}
                \overline{\alpha(e)} f(o(e)) & \text{: $e\in A_M^{(-)}$,}\\
                0 & \text{: $e\notin A_M^{(-)}$. }
                \end{cases}
        \end{align*}
We have the following Lemma. 
\begin{lemma}
\noindent
\begin{enumerate}
\item $d_{\T,M}S'=d_{O,M}$, where $(S'\psi)(e)=\psi(\bar{e})$; 
\item $d_{\T,M}^*d_{\T,M}=d_{O,M}^*d_{O,M}=\Pi_{\V_M}$, where $\V_M=\{f\in \V': v\in M \Rightarrow f(v)=0\}$. Here 
$\Pi_{\mathcal{K}}$ is the projection operator onto $\mathcal{K}\subset \V'$; 
\item $d_{\T,M}^*d_{O,M}=d_{O,M}^*d_{\T,M}$ is a self-adjoint operator $T_M'$ such that
	\[ \langle \delta_u,T_M'\delta_v \rangle
        	=\begin{cases} 
                \sum_{e: o(e)=v,t(e)=u}\overline{\alpha(\bar{e})} \alpha(e) & \text{: $u,v\in M^c$} \\
                0 & \text{: otherwise,}
        	\end{cases} \]
that is, $T_M'$ is the cut off of the self-adjoint operator of $T'$ at the target vertices $M$; 
$T_M'=\Pi_{\V_M'}T'\Pi_{\V_M'}$, where $T'$ is 
	\[ \langle \delta_u,T'\delta_v \rangle
        	= \sum_{e: o(e)=v,t(e)=u}\overline{\alpha(\bar{e})} \alpha(e). \]
\item $\sigma(T_M)\subset(-1,1)$. 
\end{enumerate}
\end{lemma}
\begin{definition}
Let $U_M': \A'\to \A'$ be the time evolution of the quantum walk induced by $T_M'$;
	\[ U'=S'C_M', \]
where $C_M'=2d_{\T,M}^*d_{\T,M}-\bs{1}_{\A'}$. 
\end{definition}
Put $U'=S'C'$ by the time evolution of the quantum walk with $M=\emptyset$. 
Then 
	\[ \langle \delta_f, U'\delta_e \rangle = \bs{1}_{\{t(e)=o(f)\}}\left(2\overline{\alpha(\bar{f})}\alpha(e)-\delta_{e,\bar{f}}\right). \]
We have 
	\[ \langle \delta_f, U_M'\delta_e \rangle
        	= 
                \begin{cases}
                \langle \delta_f, U'\delta_e \rangle &  \text{: $t(e)\notin M$,} \\
                -\delta_{e,\bar{f}}  & \text{: $t(e)\in M$.}
                \end{cases} 
        \]
\begin{lemma}It holds that 
	\begin{equation}
         d_{\T,M}^{*} (\V')\cap d_{O,M}^{*} (\V')=\bs{0}. 
        \end{equation}
\end{lemma}
\begin{proof}
Assume that $\psi\in d_{\T,M}^{*} (\V')\cap d_{O,M}^{*} (\V')\neq \bs{0}$. 
For any $\psi\in d_{\T,M}^{*} (\V')\cap d_{O,M}^{*} (\V_n)$, there exists $f\in \V'$ such that 
$d_{\T,M}^{*}f=\gamma d_{O,M}^{*}f=\gamma d_{M,\T}^{*}f$ which implies $\gamma=\pm 1$ 
since $d_{\T,M}^{*}f$ should be an eigenfunction of $S$. Taking operation $d_{\T,M}^{*}$ to both sides, we have $f=\pm T_M' f$. 
However $\sigma(T_M')\subset (-1,1)$ holds because $T_M'$ is a cut-off operator of $T'$. 
Thus $\psi\in d_{\T,M}^{*} (\V')\cap d_{O,M}^{*} (\V')=\bs{0}$. 
\end{proof}
We introduce new boundary operators as follows.  
\begin{definition}
We define $d_{\pm,M} : \A'\to \V'$ and its adjoint operator $d_{\pm,M}^{*} : \V'\to \A'$ by 
	\begin{align} 
	d_{\pm,M} &= \frac{1}{\sqrt{2(1-{T_M'}^2)}}\left( d_{\T,M}^{*}-e^{\mp \im \arccos(T_M')}d_{O,M}^{*} \right), \\
        d_{\pm,M}^{*} &= \left( d_{\T,M}^{*}-d_{O,M}^{*}e^{\pm \im \arccos(T_M')} \right)\frac{1}{\sqrt{2(1-{T_M'}^2)}}. 
        \end{align}
\end{definition}
For $f_\nu \in \mathrm{ker}(\nu-T_M')$, we have 
	\begin{equation}
        (d_{\pm,M}^{*}f_\nu)(e)=\frac{1}{\sqrt{2(1-\nu^2)}} \left\{ \overline{\alpha(\bar{e})}f_\nu(t(e))-e^{\pm \im \arccos (\nu)}\overline{\alpha(e)}f_\nu(o(e)) \right\}.
        \end{equation}
We can easily check that 
\begin{lemma}\label{decom}
\noindent
\begin{enumerate}
\item $||d_{\pm,M}^* f||=||f||$; 
\item Let $\mathcal{L}_M\subset \A'$ denote $d_{\T,M}^{*}(\V')+d_{O,M}^{*}(\V')$. Then $\mathcal{L}_M=d_{+,M}^{*}(\V')\oplus d_{-,M}^{*}(\V')$. 
\end{enumerate}
\end{lemma}
%
\begin{theorem}
It holds that
	\begin{align}
        \sigma(U'_{M}|_{\mathcal{L}_M}) &= J^{-1}(\sigma(T_M'))
        \end{align}
Moreover the generator of $U'_{M}|_{\mathcal{L}_M}$ can be expressed by 
	\[d_{+,M}(\arccos T_M') d_{+,M}^{*} \oplus d_{-,M}(-\arccos T_M')d_{-,M}^{*},\] that is, 
	\begin{align}
	U'_{M}|_{\mathcal{L}_M} &= d_{+,M}e^{\im\arccos T_M'}d_{+,M}^{*}\oplus d_{-,M}e^{-\im\arccos T_M'}d_{-,M}^{*}.
	\end{align}
\end{theorem}
\begin{proof}
By Lemma~\ref{decom}, for any $\psi\in \mathcal{L}_M$, $\psi=d_{+,M}^{*}d_{+,M}\psi \oplus d_{-,M}^{*}d_{-,M}\psi$ holds. 
We can easily check that 
	\begin{align}
        Ud_{+,M}^{*} &= d_{+,M}^{*}e^{\im \arccos(T)},  \\
        Ud_{-,M}^{*} &= d_{-,M}^{*}e^{-\im \arccos(T)}.
	\end{align}
Then we have 
	\begin{align}
        U\psi &= U(d_{+,M}^{*}d_{+,M} + d_{-,M}^{*}d_{-,M})\psi \\
        	&= \left(d_{+,M}^{*} e^{\im \arccos(T_M')}d_{+,M} + d_{-,M}^{*} e^{-\im \arccos(T_M')}d_{-,M}\right)\psi. 
        \end{align}
\end{proof}
\subsection{Tree case}
Let $\mathbb{T}=(V,A)$ be an infinite tree, and induced Hilbert spaces are $\A$ and $\V$. 
\seg{Let $V_n\subset V$ be all the vertices within the $n$-th depth that is, }
	\[ V_n=\{v\in V: \mathrm{dist}(o,v) \leq n \}.  \]
Put $\mathcal{A}_n$ and $\mathcal{A}_n^{(\pm)}\subset \ell^2(A)$ as 
	\begin{align*} 
        \mathcal{A}_n &= \spann\{\delta_e: e\in A_n\},\\ 
        \mathcal{A}_n^{(+)} &= \spann\{\delta_e: e\in A,\;t(e)\in V_n\},\; \mathcal{A}_n^{(-)} = \spann\{\delta_e: e\in A,\;o(e)\in V_n\}
        \end{align*}
and $\mathcal{V}_n\subset \ell^2(V)$ by 
	\[ \mathcal{V}_n=\spann\{\delta_v: v\in V_n\}. \]
We set the marked vertices $M$ by $V_n^c$. 
Define $d_\T^{(n)}:=d_{\T,M},d_O^{(n)}:=d_{O,M}$ by 
	\[ d_\T^{(n)}=\Pi_{\mathcal{V}_n}d_\T,\;d_O^{(n)}=\Pi_{\mathcal{V}_n}d_O. \] 
These adjoint operators are $d_\T^{(n),*}=d_\T^*\Pi_{V_n}$ and $d_O^{(n),*}=d_O^*\Pi_{V_n}$. 
The cut-off self-adjoint operator $T$ under the Dirichlet boundary condition 
outside of the $n$-th depth is denoted by $T_n$ and defined as 
	\[ T_n=\Pi_{\V_n}T\Pi_{\V_n}. \]
\seg{In this setting, we can observe the following properties on restricted boundary operators $d_{\T}$ and $d_O$ and their adjoints.}
\begin{remark}\label{remark}
\noindent
\begin{enumerate}
        \item 
        $d_\T^{(n)} =\Pi_{\V_n}d_\T=d_\T\Pi_{\A_n^{(+)}}$, $d_O^{(n)}=\Pi_{\V_n}d_O=d_O\Pi_{\A_n^{(-)}}$ \\
        or equivalently, $d_\T^{(n),*} = d_\T^*\Pi_{\V_n}=\Pi_{\A_n^{(+)}}d_\T^*$, $d_O^{(n),*}=d_O^*\Pi_{\V_n}=\Pi_{\A_n^{(-)}}d_O^*$; 
        \item $T_n=d_O^{(n)}d_\T^{(n),*}=d_\T^{(n)}d_O^{(n),*}$. 
        \end{enumerate}
\end{remark}
We define a unitary operator on $\ell^2(A)$ induced by $T_n$ as follows. 
\begin{definition}
For given $n\in \mathbb{N}$, we define  
	\begin{equation}
        U^{(n)}=S(2d_\T^{(n),*}d_\T^{(n)}-I_\A). 
        \end{equation}
\end{definition}
We call the cut-off quantum walk for this $U^{(n)}$ with the standard measurement $\mu$.
Remark that 
	\begin{equation}\label{senni}
        \langle \delta_f,U^{(n)}\delta_e\rangle 
        	= \begin{cases}
                  \langle \delta_f,U\delta_e\rangle & \text{: $t(e)\in V_n$, } \\
                  -\delta_{e,\bar{f}} & \text{: otherwise.}
        	\end{cases}
        \end{equation}
The following properties of $U^{(n)}$ can be easily seen.
\begin{remark}\label{remark2}
\noindent
\begin{enumerate}
\item $U^{(n)}$ is a unitary operator on $\mathcal{A}=\ell^2(A)$.
\item $\Pi_{\A_{n+1}}U^{(n)}=U^{(n)}\Pi_{\A_{n+1}}$.
\end{enumerate}
\end{remark}
The cut-off quantum walk $U^{(n)}$ is an approximation of $U$ in the following mean. 
\begin{proposition}
For every $\psi\in \ell^2(A)$, we have 
	\begin{equation}
        \lim_{n\to\infty}||U\psi-U^{(n)}\psi||^2_{\ell^2(A)}=0. 
        \end{equation}
\end{proposition}
\begin{proof}
By Eq.~(\ref{senni}), putting $U\psi=\phi$ and $U^{(n)}\psi=\phi_n$, 
	\begin{align}
        ||\phi-\phi_n||^2 &= \sum_{e\in A}\left| (U\psi)(e)-(U^{(n)}\psi)(e) \right|^2, \noindent \\
        	&= \sum_{e: o(e)\in V_n^c} \left| (U\psi)(e)-(-S\psi)(e) \right|^2.
	\end{align}
Since $U=2d_O^*d_T-S$, 
	\begin{align}
        ||\phi-\phi_n||^2/4 &= || (\bs{1}_\A-\Pi_{\A_n^{(-)}})d_O^*d_\T \psi ||^2_{\A}. \label{1}
	\end{align}
By Remark~\ref{remark}, 
	\begin{align}
        ||d_O^*d_\T \psi||^2 
        	&= \langle d_O^*d_\T\psi, d_O^*d_\T\psi \rangle
                =\langle d_\T\psi, d_\T\psi \rangle \noindent \\
        	&= ||d_\T\psi||^2, \label{2} \\
        ||\Pi_{\A_n^{(-)}}d_O^*d_\T \psi||^2 
        	&= \langle d_O^{(n),*}d_\T\psi, d_O^{(n),*}d_\T\psi \rangle
                =\langle d_\T\psi, \Pi_{\V_n}d_\T\psi \rangle \noindent \\
        	&= ||\Pi_{\V_n}d_\T\psi||^2. \label{3}       
        \end{align}
Combining Eqs.~(\ref{2}) and (\ref{3}) with Eq.~(\ref{1}), we have 
	\begin{equation}
        ||\phi-\phi_n||^2/4=||(\bs{1}_\V-\Pi_{\V_n})d_\T\psi||^2. 
        \end{equation}
Therefore we have $||U\psi-U^{(n)}\psi||\to 0$, $(n\to\infty)$. 
\end{proof}
The cut-off quantum walk at the $n$-th depth satisfies that $U^{(n)}(\A_{n+1})=\A_{n+1}$ and 
and it is a unitary operator on $\A_{n+1}$. 
From Remark~\ref{remark2}, we have 
	\begin{align}\label{bunkai1}
        U^{(n)} &= U^{(n)}|_{\A_{n+1}}\oplus U^{(n)}|_{\A_{n+1}^\bot}, \\
        	&= U^{(n)}|_{\A_{n+1}}\oplus (-S\Pi_{\A_{n+1}^\bot}),
        \end{align}
where $\A_{n}^\bot=\spann\{\delta_{e}:e\in A_n^c\}$. 
RHS of the second term means that $U^{(n)}$ ``freezes" the dynamics at the outside of the $n$-th depth acting as a trivial reflection operator in the same edge, 
that is, $(U^{(n)}|_{\A_{n+1}^\bot}\psi)(e)=-\psi(\bar{e})$, for every $e\in A_{n+1}^c$. 
So we focus on the first term which gives a non-trivial dynamics of the walk from now on.  
From Remark~\ref{remark2}, we have $U^{(n)}=U^{(n)}|_{\A_{n+1}}\oplus U^{(n)}|_{\A_{n+1}^\bot}$ with $U^{(n)}|_{\A_{n+1}^\bot}= -S|_{\A_{n+1}^\bot}$. 
The second term which has a non-trivial structure, $U^{(n)}|_{\A_{n+1}}$, is decomposed as follows. 
Putting for $\varphi^{(\pm)}_{u,j}\in \mathbb{C}^A$ in Definition \ref{DefFlow}, which is not necessary to be $\ell^2$-summable, $\varphi^{(\pm,n)}_{u,j}:=\Pi_{\A_n}\varphi^{(\pm)}_{u,j}$
	\[ \mathcal{F}^{(\pm)}_n=\spann\left\{ \varphi^{(\pm,n)}_{u,j}: (u,j)\in \bigcup_{v\in V_{n-1}} \left(\{v\}\times \{1,\dots,m(v)-1\}\right) \right\}. \]
\begin{theorem}
Assume that there are no leaves, $\kappa_0\geq 1$, in $\mathbb{T}$. 
It holds that
\begin{equation}
U^{(n)}|_{\A_{n+1}}=U^{(n)}|_{\mathcal{L}_n}\oplus U^{(n)}|_{\mathcal{F}_{n+1}^{(+)}}\oplus U^{(n)}|_{\mathcal{F}_{n+1}^{(-)}}. 
\end{equation}
Here 
	\begin{align}
	U^{(n)}|_{\mathcal{F}_{n}^{(\pm)}} &= \pm \Pi_{\mathcal{F}_{\pm}^{(n)}}, \\
	U^{(n)}|_{\mathcal{L}_n} &= d_+^{(n),*}e^{\im\arccos T_n}d_+^{(n)}\oplus d_-^{(n),*}e^{-\im\arccos T_n}d_-^{(n)}.
	\end{align}
\end{theorem}
\begin{proof}
The proof of the inherited part has been already done. So we will proof the birth part. 
The dimension of $\mathcal{L}_n$ is $\mathrm{dim}(\mathcal{L}_n)=2|V_n|$ while the dimension of the total state space of $\A_{n+1}$ is 
$\mathrm{dim}(\A_{n+1})=2|E_{n+1}|=2(|V_{n+1}|-1)$ which implies 
	\[\mathrm{dim}(\mathcal{L}_n^\bot \cap \A_{n+1})=2(|V_{n+1}|-|V_n|-1)>0. \]
Therefore the remaining spaces $\mathcal{L}_n^\bot \cap \A_{n+1}$ exist. 
By the way, 
	\begin{align*}
        \mathrm{dim}(\F_+^{(n)})=\mathrm{dim}(\F_-^{(n)}) &= \left| \bigcup_{u\in V_n}\{u\}\times \{1,\dots,m(u)-1\} \right|=
        	|E_{n+1}|-|V_n| \\
        	&=|V_{n+1}|-1-|V_n| \\
                &=\frac{1}{2}\mathrm{dim}(\mathcal{L}_n^\bot \cap \A_{n+1}). 
        \end{align*}
It holds that 
	\[\mathcal{L}_n^\bot \cap \A_{n+1}
        	=\left(\ker(d_\T^{(n)})\cap \mathcal{H}^{(-)} \cap \A_{n+1}\right)\oplus\left(\ker(d_\T^{(n)})\cap \mathcal{H}^{(+)} \cap \A_{n+1}\right). \]
Since $\varphi_n^{(\pm)}(u,j)\in \ker(d_\T^{(n)})\cap \mathcal{H}^{(\mp)} \cap \A_{n+1}$ and 
$\{\varphi_n^{(\pm)}(u,j)\;:\; j\in V_{n-1},\;j\in\{1,\dots,m(u)-1\}\}$ are orthogonal to each other, we obtain $\F_n^{(\pm)}=\ker(d_\T^{(n)})\cap \mathcal{H}^{(\mp)}\cap \A_{n+1}$. 
\end{proof}
Finally, we compute the density of the birth part. 
We set $B_n=|V_n|$, $\partial B_n=|V_{n+1}|-|V_n|$. 
\begin{corollary}
Assume that there are no leaves in $\mathbb{T}$. 
Let $\rho_n^{(\pm)}$ be the density of $\ker(\pm \bs{1}-U^{(n)})$, that is, $\rho_n^{(\pm)}=\dim (\ker(\pm \bs{1}-U^{(n)}))/\dim (\A_{n+1})$. 
Put $h_+:=\limsup_{n\to\infty}\partial B_n/B_n$ and $h_-:=\liminf_{n\to\infty}\partial B_n/B_n$. 
If $h_->0$, then 
	\[ 0<\frac{h_-}{2(1+h_-)}=\liminf_{n\to\infty}\rho_n^{(\pm)}\leq \limsup_{n\to\infty}\rho_n^{(\pm)}=\frac{h_+}{2(1+h_+)}\leq 1/2 \]
In particular, if $\lim_{n\to\infty}\partial B_n/B_n=:h$ exists, then 
	\[ \lim_{n\to\infty} \rho_n^{(\pm)}=\frac{h}{2(1+h)}. \]
\end{corollary}

\noindent \\
{\bf Acknowledgments}
YuH's work was supported in part by Japan Society for the
Promotion of Science Grant-in-Aid for Scientific Research (C) 25400208, (B) 24340031 and (A) 15K02055 for Challenging Exploratory Research 26610025. 
ES thanks to the financial supports of the Grant-in-Aid for Young Scientists (B) 16K17637 and 
Japan-Korea Basic Scientific Cooperation Program Non-commutative Stochastic Analysis; New Aspects of Quantum
White Noise and Quantum Walks (2015-2016). 

%


\begin{small}
\bibliographystyle{jplain}

\end{small}


\end{document}